\documentclass[11pt,draftcls,onecolumn]{IEEEtran}
\usepackage{cite}
\usepackage{graphicx}
\usepackage{setspace}
\usepackage{subfig}
\usepackage{epsfig}
\usepackage{url}
\usepackage{amsfonts,amssymb,amsmath,bm,paralist,theorem,cite,ifthen,color,amsmath}
\usepackage{array}
\usepackage{calc}
\usepackage{longtable}
\usepackage{mathrsfs}
\usepackage{epsfig}
\newtheorem{thm}{Theorem}

\theorembodyfont{\rmfamily}

\newtheorem{corollary}{Corollary}
\newenvironment{proof}[1][Proof]{\begin{trivlist}
\item[\hskip \labelsep {\bfseries #1}]}{\end{trivlist}}

\newcommand{\rank}{{\rm rank}}
\newcommand{\tr}{{\rm Tr}}

\newcommand{\bzero}{\mathbf{0}}
\newcommand{\bI}{\mathbf{I}}
\newcommand{\bU}{\mathbf{U}}
\newcommand{\bV}{\mathbf{V}}
\newcommand{\bH}{\mathbf{H}}
\newcommand{\bbH}{\bar\bH}
\newcommand{\barh}{\bar{h}}
\newcommand{\bbV}{\bar\bV}
\newcommand{\bbU}{\bar\bU}
\newcommand{\bX}{\mathbf{X}}
\newcommand{\bG}{\mathbf{G}}

\newcommand{\bx}{\mathbf{x}}

\newcommand{\bg}{\mathbf{g}}
\newcommand{\bF}{\mathbf{F}}
\newcommand{\bP}{\mathbf{P}}
\newcommand{\bQ}{\mathbf{Q}}
\newcommand{\bv}{\mathbf{v}}

\begin{document}
\title{On the Degrees of Freedom Achievable Through Interference Alignment in a MIMO Interference Channel$^*$}
%
% Single address.
% ---------------

%\author{Meisam Razaviyayn, \ Gennaday Lyubeznik, \ Zhi-Quan Luo}
\author{Meisam Razaviyayn$^\dag$\thanks{$^\dag$\;{\color{black}Department of Electrical and Computer Engineering,} University of Minnesota, Minneapolis, 55455, USA. (e-mail:
 \texttt{\{meisam,luozq\}}@umn.edu)}, Gennady Lyubeznik$^\ddag$\thanks{$^\ddag$ Department of Mathematics, University of Minnesota, Minneapolis, 55455, USA.} and Zhi-Quan Luo$^\dag$
\thanks{$^*$ This work was supported in part by the Army Research Office, grant number W911NF-09-1-0279, and in part by a research gift from Huawei Technologies Inc.}}
%\address{University of Minnesota, Minneapolis, MN 55455}
\maketitle
\begin{abstract}
Consider a $K$-user flat fading MIMO interference channel where the $k$-th transmitter (or receiver) is equipped with $M_k$ (respectively $N_k$) antennas. If an exponential (in $K$) number of generic channel extensions are used either across time or frequency, {\color{black}Cadambe and Jafar \cite{Jafar1} showed} that the total achievable degrees of freedom (DoF) can be maximized via interference alignment, resulting in a total DoF that grows linearly with $K$ even if $M_k$ and $N_k$ are bounded. In this work we consider the case where no channel extension is allowed, and establish a general condition that must be satisfied by any degrees of freedom tuple $(d_1,d_2,...,d_K)$ achievable through linear interference alignment. For a symmetric system with $M_k=M$, $N_k=N$, $d_k=d$ for all $k$, this condition implies that the total achievable DoF cannot grow linearly with $K$, and is in fact no more than $K(M+N)/(K+1)$. We also show that this bound is tight when the number of antennas at each transceiver {\color{black} is divisible by $d$, the number of data streams per user.}
%If, in addition, all users have the same DoF $d=1$, then this upper bound on the total DoF is actually tight for almost all MIMO interference channels.

%We also show that this upper bound $M+N-1$ is tight for almost all MIMO interference channels. %by example that this upper bound is tight.

%Consider a $K$ user MIMO interference fading channel where each transmitter and receiver are equipped with $M_i$ and $N_i$ antennas respectively. If channel extension is allowed and the channel response is statistically independent across frequency, the recent work \cite{Jafar1} suggests that the total achievable degrees of freedom can be maximized via interference alignment, and moreover the resulting maximum total degrees of freedom grows linearly with $K$ for bounded $M_i$ and $N_i$. In this paper we consider the case where no channel extension is allowed (i.e., the flat fading channel model), and characterize the degrees of freedom tuple achievable by any spatial interference alignment scheme for a MIMO interference channel. This work proves the conjectures in \cite{JafarDof} and provides more general results which implies that the total achievable degrees of freedom cannot grow linearly with $K$, and is in fact strictly less than $M+N$.
\end{abstract}

\section{Introduction}
\label{sec:intro}

Consider a multiuser communication system in which a number of transmitters must share common resources such as frequency, time, or space in order to send information to their respective receivers. The mathematical model for this communication scenario is the well-known \textit{interference channel}, which consists of multiple transmitters simultaneously sending messages to their intended receivers while causing interference to each other. %Interference channel is a generic model for multiuser communication and can be used in many practical applications such as Digital Subscriber Lines (DSL), Cognitive Radio (CR) systems, ad-hoc wireless
%networks, and cellular networks.

A central issue in the study of interfering multiuser systems is how to mitigate multiuser interference.
In practice, there are several commonly used methods for dealing with interference. First, we can treat the interference as noise and just focus on extracting the desired signals. This approach is widely used in practice because of its simplicity and ease of implementation, but is known to be non-capacity achieving in general. An alternative technique is channel orthogonalization whereby transmitted signals are chosen to be nonoverlapping either in time, frequency or space, leading to Time Division Multiple Access, Frequency Division Multiple Access, or Space Division Multiple Access respectively. While channel orthogonalization effectively eliminates multiuser interference, it can lead to inefficient use of communication resources and is also generally non-capacity achieving. Another interference management technique is to decode and remove interference. Specifically, when interference is strong relative to desired signals, a user can decode the interference first, then subtract it from the received signal, and finally decode its own message. Unfortunately, none of the aforementioned interference management techniques can achieve the maximum system throughput in general. %This method is less common in practice due to its complexity and security issues.

Theoretically, what is the optimal transmit/receive strategy in a MIMO interference channel? The answer is related to the characterization of the capacity region of an interference channel, i.e., determining the set of {\color{black} rate tuples} that can be achieved by the users simultaneously.
%For the noiseless case, the capacity region and the optimal precoding strategy of the two user interference channel is discussed in \cite{IC12}.
In spite of intensive research on this subject over the past three decades, % \cite{IC12}--\!\cite{IC13},
the capacity region of interference channels is still unknown (even for small number of users).
The lack of progress to characterize the capacity region of the MIMO interference channel has motivated researchers to derive various approximations of the capacity region. For example, the maximum total degrees of freedom (DoF) corresponds to the first order approximation of sum-rate capacity {\color{black}in the high SNR regime}. Specifically, in a $K$-user interference channel, we define the degrees of freedom region as the following \cite{Jafar1}:
\begin{align}\label{DoF}
\mathcal{D} = \bigg\{&(d_1,d_2,\ldots,d_K) \in \mathbb{R}_+^K\mid \forall  (w_1,w_2,\ldots,w_K) \in \mathbb{R}_+^K,\nonumber\\
&\sum_{k=1}^K w_k d_k \leq  \limsup_{{\rm SNR} \rightarrow \infty} \left[\sup_{\mathbf{R} \in \mathcal{C}} \frac{1}{\log{\rm SNR}} \sum_{k=1}^K w_kR_k\right] \bigg\}, %\nonumber
\end{align}
where $\mathcal{C}$ is the capacity region and $R_k$ is the rate of user~$k$.
We can further define the total DoF in the system as the following:
\begin{align}
\eta = \max_{(d_1,d_2, \ldots,d_K) \in \mathcal{D}} d_1 + d_2 + \ldots + d_K. \nonumber
\end{align}
Intuitively, the total DoF is the number of independent data streams that we can communicate interference-free in the channel. % at high SNR regime.

It is well known that for a point-to-point MIMO channel with $M$ antennas at the transmitter and $N$~antennas at the receiver, the total DoF is $\eta = \min \{M,N\}$. Different approaches such as SVD precoder or V-BLAST can be used to achieve this DoF bound. For a 2-user MIMO fading interference channel with user $k$  equipped with $M_k$ transmit antennas and $N_k$ receive antennas ($k=1,2$), Jafar and Fakhereddin  \cite{JafarFakhK2} proved that the maximum total DoF is % given by
\begin{align}
\eta = \min \left\{M_1 + M_2, N_1 + N_2 , \max\{M_1,N_2\}, \max\{M_2,N_1\}\right\}. \nonumber
\end{align}
This result shows that for the case of $M_1 = M_2 = N_1 = N_2$, the total DoF in the system is the same as the single user case. In other words, we do not gain more DoF by increasing the number of users from one to two. Interestingly, if {\color{black} generic} channel extensions (drawn from a continuous probability distribution) are allowed either across time or frequency, Cadambe and Jafar  \cite{Jafar1} showed that the total DoF is $\eta = {KM}/{2}$ for a $K$-user MIMO interference channel, where $M$ is the number of transmit/receive antennas per user. {\color{black}This  result implies that each user can effectively utilize half of the total system resources in an interference-free manner by {aligning the interference at all receivers}}\footnote{\color{black}The idea of interference alignment was introduced in \cite{waterloo,JafarXfirst, Birk} and the terminology ``interference alignment" was first used in \cite{IATerminology}.}. The principal assumption enabling this surprising result is that the channel extensions are exponentially long in $K$ and are {\color{black} generic} (e.g., drawn from a continuous probability distribution). If channel extensions are restricted to have a polynomial length or are not {\color{black} generic}, the total DoF for a MIMO interference channel is still largely unknown even for the Single-Input-Single-Output (SISO) interference channel. For the 3-user special case, reference \cite{BT} provided a characterization of the total achievable DoF as a function of the diversity.
%H{\o}st, Madsen and Nosratinia conjectured in \cite{HMNosratinia} that the total number of DoF in a SISO fading interference channel is $\eta=1$. This conjecture was later disproved by Cadambe \textit{et al.} \cite{RejectHMNosratinia} who used a \emph{nonlinear signal level alignment} scheme to show that the total DoF is at least 1.2. %However, when restricted to the class of linear interference alignment schemes, the total DoF for a SISO fading interference channel is still unknown if channel extensions are not allowed.
{\color{black} In the absence of channel extensions, the computational complexity of numerically designing an interference alignment scheme has been shown to be NP-hard \cite{IANPhard} in the number of users}.

The main theoretical investigation pertaining to the current work is \cite{Jafarbound} by Yetis \textit {et\ al.\ }who studied the maximum achievable DoF for a MIMO interference channel without channel extension. In general, linear interference alignment can be described by a set of {\color{black} bilinear equations} which correspond to the zero-forcing conditions at each receiver. For a $K$-user system, there are a total of $K(K-1)$ such coupled quadratic matrix equations whose unknowns are the transmit/receive beamforming matrices to be designed. Moreover, the achievability of a given tuple of DoF corresponds to these quadratic equations having a solution (in the form of beamforming matrices) whose individual matrix ranks are given by the DoFs. One can easily count the number of ``independent unknowns" and the number of scalar equations in this quadratic system defining interference alignment. It is then tempting to conjecture, as was done in \cite{Jafarbound}, that the interference alignment is feasible if and only if the number of equations is no more than the number of unknowns in each subsystem of the quadratic equations. When the latter is true, the authors of \cite{Jafarbound} called the corresponding system \emph{proper}.
%It was conjectured \cite{Jafarbound} that the achievability of a given set of DoF is equivalent to the properness of the quadratic polynomial system characterizing the corresponding interference alignment conditions.
However, except for some special cases involving a small number of users and antennas, the investigation of \cite{Jafarbound} was largely inconclusive.

In this paper, we settle the conjecture of \cite{Jafarbound} completely in one direction, and partially in the other. In particular, we consider the case where no channel extension is allowed, and use results from the field theory to establish a general condition that must be satisfied by any DoF tuple achievable through linear interference alignment. This condition shows that the improperness property (in the sense of \cite{Jafarbound}) indeed implies the infeasibility of interference alignment. For the symmetric system with $M_k=M$ and $N_k=N$ for all $k$, this condition implies that the total achievable DoF cannot grow linearly with the number of users, and is in fact no more than $M+N-1$. This is in sharp contrast to the case with independent channel extensions for which the total DoF can grow linearly with the number of users. For the converse direction, we show that if  all users have the same DoF $d$ and the number of antennas $M_k$, $N_k$ are divisible by $d$ for each $k$, then the properness of the quadratic system implies the feasibility of interference alignment for {\color{black} generic choice of channel coefficients (e.g., drawn from a continuous probability distribution)}. If in addition, $M_k=M$ and $N_k=N$ for all $k$ and $M, N$ are divisible by $d$, then our results imply that interference alignment is achievable if and only if $(M+N)\ge d(K+1)$. In the simulation section, we use these established DoF bounds to numerically benchmark the performance of several existing algorithms for interference alignment and sum-rate maximization.

%Finally, we establish some condition under which the interference alignment is feasible for almost all channel realizations.

%However, it has been shown that for $K$~user channel ($K > 2$), where each user (transmitter or receiver) is equipped with $M$~antennas, the total DoF of $\eta = \frac{KM}{2}$ is achievable. To prove this result, authors in \cite{Jafar1} use i.i.d. fading assumption with infinite diversity. On the other hand, for channels with limited amount of diversity, the total DoF is still unknown even for the simple case of constant channel coefficients. Host Madsen and Nosratinia conjectured in \cite{HMNosratinia} that the total number of DoF in the constant interference channel is one. However, this conjecture has been rejected by Cadambe \textit{et al.} in \cite{RejectHMNosratinia}. The total DoF is still unknown for the constant interference channel.
%
%Although the optimum strategy for maximizing the total DoF in a system with finite diversity is still unknown for the general case, one promising approach for maximizing the total DoF in the system is \textit{spatial interference alignment} \cite{Jafar2}. In this paper, we consider the spatial interference alignment problem and derive an upper bound for the total achievable DoF in the system via this scheme. Our result provides a rigrous proof for the conjectures in \cite{JafarDof}.

\section{System Model}
\label{sec:System_Model}

Consider a MIMO interference network consisting of $K$ transmitter~- receiver pairs, with transmitter $k$ sending $d_k$ independent data streams to receiver $k$. Let $\mathbf{H}_{kj}$ be an $M_j \times N_k$ matrix that represents the channel gain matrix from transmitter~$j$ to receiver~$k$ where $M_j$ and $N_k$ denote the number of antennas at transmitter~$j$ and receiver~$k$, respectively. The received signal at receiver $k$ is given by
\begin{align}
 \mathbf{y}_k = \sum_{j=1}^{K} \mathbf{H}_{kj} \mathbf{x}_j
+ \mathbf{n}_k\nonumber
\end{align}
\normalsize where $\mathbf{x}_j$ is an $M_j \times 1$ random vector that represents the transmitted signal of user $j$ and $\mathbf{n}_k \sim \mathcal{N} (\mathbf{0}, \sigma^2 \mathbf{I})$ is a zero mean additive white Gaussian noise.\\

Throughout this paper, we focus on {\em linear} transmit and receive strategies that can maximize system throughput. In this case, transmitter~$k$ uses a beamforming matrix $\mathbf{V}_k$ in order to send a signal vector $\mathbf{s}_k$ to its intended receiver~$k$. On the other side, receiver~$k$ estimates the transmitted data vector $\mathbf{s}_k$ by using a linear beamforming matrix~$\mathbf{U}_k$, i.e.,
\begin{align}
 \mathbf{x}_k = \mathbf{V}_k \; \mathbf{s}_k, \;\;\;\;\;\;\;
\hat{\mathbf{s}}_k = \mathbf{U}_k^H \mathbf{y}_k\nonumber
\end{align}
%\normalsize
where the power of the data vector $\mathbf{s}_k \in \mathbb{R}^{d_k \times 1}$ is normalized such that $E[ \mathbf{s}_k \mathbf{s}_k^H ] = \mathbf{I}$, and $\hat{\mathbf{s}}_k$ is the estimate of $\mathbf{s}_k$ at the $k$-th receiver. The matrices $\mathbf{V}_k \in \mathbb{C}^{M_k \times d_k}$ and $\mathbf{U}_k \in \mathbb{C}^{N_k \times d_k}$ are the beamforming matrices at the $k$-th transmitter and receiver respectively.  %It is known that the problem of designing optimal beamformers to maximize sum rate of the system is NP-hard \cite{complexity} even in the single transmit/receive antenna case. Notice that recent works \cite{Jafar1,Jafar2} suggest that the optimal strategy should have interference alignment structure in the high SNR regime. Therefore, we are led to find a linear transmission-reception strategy that can maximize the total degrees of freedom. %Notice that the
Without channel extension, the linear interference alignment conditions can be described by the following zero-forcing conditions \cite{Jafarbound,IANPhard}
\begin{align}
&\mathbf{U}_k^H\mathbf{H}_{kj}\mathbf{V}_j=\bzero,\;\;\;\;\  \forall\ j\neq k, \label{IA1}\\
&\rank\left(\mathbf{U}_k^H \mathbf{H}_{kk}\mathbf{V}_k\right)=d_k,\quad \forall\ k. \label{IA2}
\end{align}
%\normalsize
The first equation guarantees that all the interfering signals
at receiver $k$ lie in the subspace orthogonal to $\mathbf{U}_k$, while
the second one assures that the signal subspace $\mathbf{H}_{kk}
\mathbf{V}_k$ has dimension $d_k$ and is linearly independent of
the interference subspace. Intuitively, as the number of users $K$ increases, the number of constraints on the beamformers $\{\bU_k,\bV_k\}$ increases quadratically in $K$, while the number of design variables in $\{\bU_k,\bV_k\}$ only increases linearly. This suggests the above interference alignment can not have a solution unless $K$ or $d_k$ is small.

%In the absence of channel noise (or equivalently in the high SNR regime),
The interference alignment conditions \eqref{IA1} and \eqref{IA2} imply that each transmitter $k$ can use a linear transmit/receive strategy to communicate $d_k$ interference-free independent data streams to receiver $k$ (per channel use). In this case, it can be checked that $d_k$ represents the DoF achieved by the $k$-th transmitter/receiver pair in the information theoretic sense of \eqref{DoF}. In other words, the vector $(d_1,d_2,...,d_K)$ in \eqref{IA1} and \eqref{IA2} represents the tuple of DoF achieved by linear interference alignment. Intuitively, the larger the values of $d_1,d_2$,...,$d_K$, the more difficult it is to satisfy the interference alignment conditions \eqref{IA1} and \eqref{IA2}.

%For simplicity of notations, when we use the term ``degrees of freedom" or ``degrees of freedom tuple", we mean the degrees of freedom (tuple) that can be achieved via spatial interference alignment by solving (1)-(2).

\section{Bounding the Total DoF Achievable via Linear Interference Alignment}

Our goal is to study the solvability of the interference alignment problem \eqref{IA1}-\eqref{IA2} and derive a general condition that must be satisfied by any DoF tuple $(d_1,d_2,...,d_K)$ achievable through linear interference alignment for generic choice of channel matrices. We will also provide some conditions under which this upper bound is achievable.

Let us denote the polynomial equations in \eqref{IA2} by the index set
$$\mathcal{J} \triangleq \{(k,j)\mid 1\le k\neq j\le K\}.$$
The following theorem provides an upper bound on the total achievable DoF when no channel extension is allowed.
\begin{thm} \label{thm:bd}
Consider a $K$-user flat fading MIMO interference channel where the channel matrices $\{\mathbf{H}_{ij}\}_{i,j=1}^K$ are generic ({\color{black}e.g., drawn from a continuous probability distribution}). Assume no channel extension is allowed. Then any tuple of degrees of freedom $(d_1,d_2,...,d_K)$ that is achievable through linear interference alignment \eqref{IA1} and \eqref{IA2} must satisfy the following inequalities
\begin{align}
 &\min \{M_k, N_k\} \geq d_k  , \quad \forall\; k, \label{I1} \\
 &\max \{M_k,N_j\} \geq d_k+d_j ,  \quad \forall \;k,j, k\neq j, \label{I2}\\
 &\sum_{k: (k,j)\in \mathcal{I}} (M_k - d_k)d_k + \sum_{j: (k,j)\in \mathcal{I}} (N_j - d_j)d_j \ge \sum_{(k,j)\in\mathcal{I}} d_k d_j,\quad \forall \; \mathcal{I}  \subseteq \mathcal{J}.\label{IAbound}
\end{align}
\end{thm}

Condition~\eqref{IAbound} in Theorem~\ref{thm:bd} can be used to bound the total DoF achievable in a MIMO interference channel. The following corollary is immediate.
\begin{corollary}\label{co1}
Assume the setting of Theorem~\ref{thm:bd}. Then the following upper bounds hold true.
\begin{enumerate}
\item [(a)] In the case of $d_k=d$ for all $k$, interference alignment is impossible unless
\[
d\le\frac{1}{K(K+1)}\sum_{k=1}^K(M_k+N_k).
\]
\item [(b)] In the case of $M_k+N_k=M+N$, interference alignment requires
\[
\left(\sum_{k=1}^Kd_k\right)^2+\sum_{k=1}^Kd_k^2\le (M+N)\sum_{k=1}^Kd_k
\]
which further implies
\[
\sum_{k=1}^Kd_k < (M+N).
\]
\end{enumerate}
\end{corollary}

Part (b) of Corollary~\ref{co1} shows that the total achievable DoF in a MIMO interference channel is bounded by a constant $M+N-1$, regardless of how many users are present in the system. While this bound is an improvement over the single user case which has a maximum DoF of $\min\{M,N\}$, it is significantly weaker than the maximum achievable total DoF for a diagonal frequency selective (or time varying) interference channel. The latter grows linearly with the number of users in the system \cite{Jafar1}.

The rest of this section is devoted to the proof of Theorem~\ref{thm:bd} and its converse. Since we will use several concepts and results from the field theory \cite{Field_book} and algebraic geometry \cite{Matsumura,Atiyah}, we first provide a brief review of the necessary algebraic background.

\subsection{Algebraic Preliminaries}\label{sec:preliminary}
%\appendix{Algebraic Dependence and Transcendence Basis:}

Let $\mathcal{K}, \mathcal{F}$ be two fields such that $\mathcal{K} \subseteq \mathcal{F}$. In this case, we say $\cal F$ is an extension of $\cal K$, denoted by ${\cal F}/{\cal K}$. Let us use $\mathcal{K}[z_1,z_2,\ldots,z_n]$ to denote the ring of polynomials with coefficients drawn from $\cal K$. We say~$\alpha_1, \alpha_2,\ldots, \alpha_n \in \mathcal{F}$ are \textit{algebraically dependent} over~$\mathcal{K}$ if there exists a nonzero polynomial~$f(z_1,z_2,\ldots,z_n) \in \mathcal{K}[z_1,z_2,\ldots,z_n]$ such that
\begin{align}
\label{AlgInd}
f(\alpha_1,\alpha_2,\ldots,\alpha_n) =0.
\end{align}
{\color{black}Otherwise, we say that they are \textit{algebraically independent} over~$\mathcal{K}$. The largest cardinality of an algebraically independent set is called the transcendence degree of ${\cal F}$ over $\cal K$. An element $\alpha\in \cal F$ is said to be algebraic over $\cal K$ if there exists a nonzero polynomial $f\in \mathcal{K}[z]$ such that $f(\alpha)=0$; else, we say $\alpha$ is transcendental over $\cal K$.}
 \\

%If every element of $\cal F$ except those in $\cal K$ is transcendental over $\cal K$, then we say ${\cal F}/{\cal K}$ is a pure transcendental extension of $\cal K$. Let ${\cal K}(z_1,z_2,...,z_n)$ denote the field of rational functions in $n$ variables $(z_1,z_2,...,z_n)$ with coefficients taken from $\cal K$. Then ${\cal K}(z_1,z_2,...,z_n)/{\cal K}$ is a pure transcendental extension of $\cal K$. % with transcendence degree of $n$.

\noindent
\textbf{Example 1.}
Let $\mathcal{K} = \mathbb{C}$ be the field of complex numbers and $\mathcal{F}=\mathbb{C}(x_1,x_2)$  be the field of rational functions in variables~$x_1,x_2$. Then, the polynomials
\begin{align}
g_1 = x_1^2 x_2, \quad g_2 = x_2^2, \quad g_3 = x_1x_2 \nonumber
\end{align}
are algebraically dependent over $\mathbb{C}$ because $f(g_1,g_2,g_3)=0$ identically for all $(x_1,x_2)$, where $f(z_1,z_2,z_3) = z_1^2 z_2 - z_3^4$.\\

\noindent
\textbf{Example 2.} The two complex numbers $a = \sqrt{\pi}, \; b = 3 \pi + 2$ are algebraically dependent over the field of rational numbers because by defining $f(z_1 , z_2) = 3z_1^2 - z_2 + 2$, we have $f(a,b) = 0$.\\

Notice that the definition of algebraic independence is in many ways similar to the standard notion of linear independence from linear algebra. In fact, if the function~$f$ in \eqref{AlgInd} is required to be linear, then algebraic independence reduces to the usual concept of linear independence. Similar to linear algebra, we can define a basis for the field~$\mathcal{F}$ using the notion of algebraic independence. In particular, %a basis of $\cal F$ over the subfield $\cal K$ consists of a maximal number of algebraically independent elements in $\cal F$. The cardinality of any basis is called the dimension of $\cal F$ over $\cal K$.
given any algebraically independent set $S$ over the field $\cal K$, let ${\cal K}(S)$ denote the field of rational functions in $S$ with coefficients taken from the field $\cal K$.
%Then ${\cal K}(S)$ is a pure transcendental extension of $\cal K$, in the sense that every element in ${\cal K}(S)$ except those in $\cal K$ is a transcendental element over $\cal K$.
For any field extension ${\cal F}/{\cal K}$, it is always possible to find a set $S$ in ${\cal F}$, algebraically independent over $\cal K$, such that ${\cal F}$ is an algebraic extension of ${\cal K}(S)$. Such a set $S$ is called a transcendence basis of $\cal F$ over $\cal K$. All transcendence bases have the same cardinality, equal to the transcendence degree of the extension ${\cal F}/{\cal K}$. {\color{black}If every element in $\cal F$ is algebraic over $\cal K$, then we say ${\cal F}/{\cal K}$ is an algebraic extension. In this case, the transcendence degree of ${\cal F}$ over $\cal K$ is zero.} \\

\noindent
\textbf{Example 3.} The two polynomials $g_1$ and $g_2$ in Example 1 are algebraically independent over $\mathbb{C}$. Together, they constitute a transcendental basis for $\mathbb{C}(x_1,x_2)$ over $\mathbb{C}$.
\\

The following table shows similar concepts between linear algebra and transcendental field extension (see \cite{Field_book,Atiyah} for more details).
\begin{center}
\begin{tabular}{|c|c|}
  \hline
  % after \\: \hline or \cline{col1-col2} \cline{col3-col4} ...
  \textbf{Linear algebra} & \textbf{Transcendental field extension} \\
  \hline
  linear independence & algebraic independence \\
  \hline
  $A \subseteq {\rm span}(B)$ & $A$ algebraically dependent on $B$ \\
  \hline
  linear basis & transcendence basis \\
  \hline
  dimension & transcendence degree \\
  \hline
\end{tabular}
\end{center}
\medskip

In linear algebra, it is well known that any $(n+1)$ vectors $\bv_1,\bv_2,...,\bv_{n+1}$ in an $n$-dimensional vector space must be linearly dependent. In other words, there exists a nonzero linear function $f(z_1,z_2,...,z_{n+1})$ such that $f(\bv_1,\bv_2,...,\bv_{n+1})=0$. A similar result holds for algebraic independence. For example, any $(n+1)$ polynomials $g_1$, $g_2$,..., $g_{n+1}$ defined on $n$ variables $(x_1,x_2,...,x_n)$ must be algebraically dependent. Consequently, there exists a nonzero polynomial $f(z_1,z_2,...,z_{n+1})$ such that
$$
f(g_1,g_2,...,g_{n+1})=0,\quad \forall\ (x_1,x_2,...,x_n).
$$
Example 1 is an instance of this property with $n=2$. {\color{black} The following example states this property, to be used in the proof of Theorem~\ref{thm:bd}, in a more formal setting.}\\

\noindent
\textbf{Example 4.} {\color{black}Let~$\mathbb{C}(z_1,z_2,\ldots,z_n)$ denote the field of rational functions in~$n$ variables with coefficients in~$\mathbb{C}$. The set~$\{z_1,z_2,\ldots,z_n\}$ is a maximal algebraically independent set in~$\mathbb{C}(z_1,z_2,\ldots,z_n)$. Hence the transcendence degree of the field extension $\mathbb{C}(z_1,z_2,\ldots,z_n)/\mathbb{C}$ is~$n$. Furthermore, for any $m$~polynomials
\[g_1(z_1,z_2,\ldots,z_n), \;g_2(z_1,z_2,\ldots,z_n), \;\ldots, \; g_m(z_1,z_2,\ldots,z_n),\]
where $m>n$, there exists a nonzero polynomial $f(\cdot)$ such that~$f(g_1,g_2,\ldots,g_m)=0, \ \forall \; z_1,z_2,\ldots, z_n.$}
\\

{\color{black}
Next we describe a useful local expansion of a multivariate polynomial function. Recall that for any univariate polynomial $f$ and any $\bar x\in\mathbb{C}$, there holds
\[
f(x)=f(\bar x)+(x-\bar x)g(x),\ \mbox{for all $x\in \mathbb{C}$,}
\]
where $g$ is some polynomial dependent on $\bar x$ and the coefficients of $f$ only. Similarly, for a $n$-variate polynomial $f$ defined on the variables $\bx=(x_1,x_2,...,x_n)$ and any $\bar\bx\in\mathbb{C}^n$, we have
\[
f(\bx)=f(\bar \bx)+\sum_{i=1}^n(x_i-\bar x_i)g_i(\bx)=f(\bar \bx)+(\bx-\bar \bx)^T\bg(\bx),\ \forall\; \bx\in \mathbb{C}^n,
\]
where each $g_i$ is some polynomial dependent on $\bar \bx$ and the coefficients of $f$ only. If we replace the scalar variable $x_i$ by a matrix variable $\bX_i$, then we can write
\begin{equation}\label{eq:taylor}
f(\bX)=f(\bar \bX)+\sum_{i=1}^n\tr\left((\bX_i-\bar \bX_i)\bG_i(\bX)\right),\ \forall\; \bX,
\end{equation}
where each $\bG_i$ is a matrix whose entries are polynomials dependent on the entries of $\bar\bX$ and the coefficients of $f$ only. The local expansion \eqref{eq:taylor} will be used in the proof of Theorem~\ref{thm:bd}.
}

To prove the converse of Theorem~\ref{thm:bd}, we will use the concepts of Zariski topology and a Zariski constructible set. We briefly review these concepts next (see \cite{Matsumura} for more details). Consider $\mathbb{C}^n$, the $n$-dimensional vector space over the field of complex numbers $\mathbb{C}$. [One can replace $\mathbb{C}$ by any algebraically closed field.] The Zariski topology for $\mathbb{C}^n$ is defined by specifying its closed sets, and these are taken simply to be all the algebraic sets in $\mathbb{C}^n$. That is, the closed sets under Zariski topology are those of the form
\[
S=\{\bx\in\mathbb{C}^n\mid f_i(\bx)=0,\ i=1,2,...,m\}
\]
where $\{f_i\}_{i=1}^m$ is any set if polynomials with coefficients taken from $\mathbb{C}$. For example, the entire space $\mathbb{C}^n$ is Zariski closed (Take~$m=1$ and $f_1$ to be the zero function, i.e., $f_1(x) = 0,\ \forall \; x$). All other Zariski closed sets have zero measure. A nonempty Zariski open set (the complement of a Zariski closed set) always has dimension $n$. If a property holds over a Zariski open set, we say the property holds \emph{generically}.

In topology, a set is locally closed if it is the intersection of an open set with a closed set. A \emph{constructible set} is defined as a finite union of locally closed sets. Thus, a Zariski constructible set is simply a finite collection of sets, each defined by the feasible set of finitely many polynomial equations and polynomial inequalities. Clearly, if a Zariski constructible set has dimension $n$, then it must contain a Zariski open subset.\\

Let $\phi_1,\phi_2,\ldots, \phi_n$ be polynomials in $x_1,x_2, \ldots,x_n$ with coefficients from~$\mathbb{C}$. They define a map~$\Phi: \mathbb{C}^n \mapsto \mathbb{C}^n$ as follows: $\Phi(\bx) = (\phi_1(\bx),\phi_2(\bx),\ldots, \phi_n(\bx)) \in \mathbb{C}^n$. Chevalley's Theorem says that the image of this map is a constructible set (see \cite{Atiyah} for more details).\\

\noindent
\textbf{Example 5.} Let $\Phi: \mathbb{C}^2 \mapsto \mathbb{C}^2$ be defined by $\Phi(\bx) = (\phi_1(\bx), \phi_2 (\bx))$ where~$\phi_1(x) = x_1$ and $\phi_2(x) = x_1 x_2$. Let $\mathcal{L}$ be the line $\{\bx \in \mathbb{C}^2: x_1 =0\}$. The image of~$\Phi$ is the union of two locally closed sets, $\mathbb{C}^2 \backslash \mathcal{L}$ (which is in fact open) and  the point~$(0,0)$ (which is indeed closed).\\

Let the image of $\Phi$ be the union of locally closed subsets~$\mathcal{W}_1, \mathcal{W}_2, \ldots, \mathcal{W}_p$ where~$\mathcal{W}_i = \mathcal{U}_i \bigcap \mathcal{V}_i$ and $\mathcal{V}_i$ is closed and $\mathcal{U}_i$ is open. Assume the Jacobian of $\phi_1,\phi_2,\ldots,\phi_n$ is nonsingular at some point~$\bx \in \mathbb{C}^n$. The Implicit Function Theorem says that the image of $\Phi$ contains a small open disc around $\Phi(\bx)$, hence the measure of the image is nonzero. This implies that for some~$i$, $\mathcal{V}_i = \mathbb{C}^n$ and $\mathcal{W}_i = \mathcal{U}_i$, i.e., the image of the map~$\Phi(\cdot)$ contains a Zariski open set. {\color{black} Thus, if a certain property is shown to hold over the image of a polynomial map $\Phi:\mathbb{C}^n\mapsto \mathbb{C}^n$ whose Jacobian is nonsingular at some point, then this property must hold generically. We will use this approach to establish the generic feasibility of interference alignment for certain MIMO interference channels (Theorem~\ref{thm:achievability}).}

\subsection{Proof of Theorem~\ref{thm:bd}}

We now use the transcendental field extension theory to establish Theorem~\ref{thm:bd}.
\begin{proof}
The inequality \eqref{I1} is obvious due to \eqref{IA2}. To prove \eqref{I2}, assume $M_j \leq N_k$. Since $\bH_{kj}$ is generic, $\rank(\bH_{kj}\bV_j) = d_j$. Furthermore, due to \eqref{IA2}, the beamformer~$\bU_k$ must be full rank and hence $d_k + d_j$ must be no more than the total dimension~$N_k$. Similar argument shows that $d_k + d_j \leq M_j$ when~$M_j \geq N_k$. Thus, $d_k + d_j \leq \max\{M_j , N_k\}$.\\

For simplicity of notations, we prove \eqref{IAbound} for the case~$\mathcal{I} = \mathcal{J}$. When~$\mathcal{I} \subset \mathcal{J}$, the proof is the same except that we need to focus on a subset of equations/variables. Now, we prove \eqref{IAbound} for the case of~$\mathcal{I} = \mathcal{J}$ by contradiction. Assume the contrary that
\begin{align}
\sum_{k=1}^K (M_k - d_k)d_k + \sum_{j = 1}^K (N_j - d_j)d_j < \sum_{k,j = 1, k\neq j}^K d_k d_j,
\end{align}
and the interference alignment conditions in \eqref{IA1} and \eqref{IA2} are satisfied. The interference alignment condition \eqref{IA2} implies that
$\bU_k$ and $\bV_k$ must have full column rank.
{\color{black}By applying appropriate linear transformations to the rows of $\bU_k$ and $\bV_k$, we can write
\begin{equation}\label{eq1}
\bU_k=\bP_k^u\left[\begin{array}{l} \bI \\ \bar\bU_k\end{array}\right] \bQ_k^u,\quad\bV_k=\bP_k^v\left[\begin{array}{l} \bI \\ \bar\bV_k\end{array}\right]\bQ_k^v, \quad \forall k,
\end{equation}
where $\bar\bU_k$ and $\bar\bV_k$ are some matrices of size $(N_k-d_k)\times d_k$ and $(M_k-d_k)\times d_k$ respectively. The matrices $\bP_k^u$ and $\bP_k^v$ are square permutation matrices of size $N_k\times N_k$ and $M_k\times M_k$ respectively, while $\bQ_k^u, \bQ_k^v$ are some invertible matrices of size~$d_k \times d_k$.
Define $\bar\bH_{ij}=\bP_i^{u \;-1}\bH_{ij}\bP^{v \;-1}_j$ to be the permuted version of~$\bH_{kj}$.}
%In fact, any beamforming matrices of the form \eqref{eq1} satisfy the second part of the interference alignment condition \eqref{IA} with probability 1.
We can partition the matrix $\bar\bH_{kj}$ as
\[
\bar\bH_{kj}=\left[\begin{array}{cc}
\bar\bH_{kj}^{(1)}& \bar\bH_{kj}^{(2)} \\ [5pt] \bar\bH_{kj}^{(3)} & \bar\bH_{kj}^{(4)} \end{array}\right]
\]
where $\bar\bH_{kj}^{(1)}$ is of size $d_k \times d_j$. {\color{black}Since the channel matrices $\{\bH_{kj}\}_{k\neq j}$
are drawn from a continuous probability distribution, the transformed
channel matrices $\{\bar\bH^{(1)}_{kj}\}_{k\neq j}$ remain generic.}
{\color{black}Rewriting the linear interference alignment condition \eqref{IA1} in terms of $\bar\bU_k$ and $\bar\bV_k$, we obtain
\begin{equation}
\label{eq1half}
\left[\begin{array}{cc}\bI & \bar\bU_k^H\end{array}\right]\left[\begin{array}{cc}
\bar\bH_{kj}^{(1)}& \bar\bH_{kj}^{(2)} \\ [5pt] \bar\bH_{kj}^{(3)} & \bar\bH_{kj}^{(4)} \end{array}\right]\left[\begin{array}{l} \bI \\ \bar\bV_j\end{array}\right]=\bzero
\end{equation}
or equivalently
\begin{equation}\label{eq2}
\bar\bH_{kj}^{(1)}+\bar\bU^H_k\bar\bH_{kj}^{(3)}+\bar\bH_{kj}^{(2)}\bar\bV_j+\bar\bU^H_k\bH_{kj}^{(4)}\bar
\bV_j=\bzero,\quad \forall \ j\neq k.
\end{equation}
}
%for all  $j\neq k$.
The above system of quadratic equations, {\color{black}first derived in \cite{Jafarbound}}, is equivalent to the interference alignment condition \eqref{IA1}. The number of scalar equations in \eqref{eq2} is
\[\sum_{j,k=1,j\neq k}^Kd_kd_j, \]
%=(K-1)\sum_{k=1}^Kd_k.\]
while the total number of scalar variables (i.e., the scalar entries of the unknown matrices $\{\bar\bU_k\}$'s and $\{\bar\bV_k\}$'s) is
\[
\sum_{k=1}^K(M_k-d_k)d_k+\sum_{k=1}^K(N_k-d_k)d_k=\sum_{k=1}^K(M_k+N_k-2d_k)d_k.
\]
So if
\begin{equation}\label{eq3}
\sum_{k=1}^K(M_k+N_k-2d_k)d_k<\sum_{j,k=1,j\neq k}^Kd_kd_j,
\end{equation}
then we would have more constraints than unknowns in the interference alignment condition \eqref{eq2},
%The condition \eqref{eq3} is equivalent to
%\begin{align}
%\sum_{k=1}^K(M_k+N_k)d_k &<\sum_{j,k=1,j\neq k}^Kd_kd_j+2\sum_{k=1}^Kd_k^2\nonumber\\
%&=\left(\sum_{k=1}^Kd_k\right)^2+\sum_{k=1}^Kd_k^2\nonumber
%\end{align}
which we will argue cannot hold. %, thus establishing \eqref{contrary}.

Let us consider the field $\cal F$ defined over the field of complex numbers $\mathbb{C}$, consisting of all rational functions in the entries of the matrices $\{\bar\bU_k\}_{k=1}^K$ and $\{\bar\bV_k\}_{k=1}^K$. Note that the entries of the matrices $\{\bar\bU_k,\bar\bV_k \}_{k=1}^K$ form a transcendence basis for $\cal F$ over $\mathbb{C}$. Thus, the transcendence degree of $\cal F$ is  $\sum_{k=1}^K(M_k+N_k-2d_k)d_k$, which is equal to the number of entries in the matrices $\{\bar\bU_k,\bar\bV_k \}_{k=1}^K$.

Now, let us consider the matrices $\bH_{kj}^{(2)},  \bH_{kj}^{(3)}, \bH_{kj}^{(4)}$ for all $k,j,k \neq j$ and define the matrix $\bF_{kj}$:
\begin{equation}\label{eq:F}
\bF_{kj}(\bbU,\bbV) \triangleq -\left(\bar\bU_k^H \bar\bH_{kj}^{(3)}  + \bar\bH_{kj}^{(2)}\bar\bV_j + \bar\bU_k^H \bar\bH_{kj}^{(4)} \bar\bV_j\right),
\end{equation}
\noindent for all $k,j$ with $k \neq j$.
Note that $\bF_{kj}$ is a $d_k \times d_j$ matrix, with each entry being a quadratic polynomial function of the entries in the matrices $\bar\bU_k$ and $\bar\bV_k$. As a result, the entries of $\bF_{kj}$ belong to the field $\cal F$. %Let us define $f_{k,j}^1, f_{k,j}^2, \ldots, f_{k,j}^{d_k d_j}$ to be the polynomials in $\bF_{kj}$. Consider the set of polynomials $\mathcal{F}$ where
%\[
%\mathcal{F} \triangleq \left\{ f_{kj}^{\ell}|\; k,j = 1,2,\ldots, K,\; k\neq j,\; \ell = 1,2,\ldots, d_kd_j\right\}
%\]
Moreover, if \eqref{eq3} holds, then the number of quadratic polynomials given in the matrices $\{\bF_{kj}\}_{k\neq j}$ is strictly larger than the transcendence degree of $\cal F$ over $\mathbb{C}$. Hence, as we discussed in the algebraic preliminaries (Section~\ref{sec:preliminary}; see also \cite[Chapter 8]{Field_book}), these quadratic polynomials in $\mathcal{F}$ must be algebraically dependent. This implies that there exists a nonzero polynomial $p$ which vanishes at the quadratic polynomials corresponding to the entries of the matrices $\{\bF_{kj}\}_{k\neq j}$, i.e.,
\begin{align}
p\left(\bF_{12}(\bbU,\bbV),\bF_{13}(\bbU,\bbV),\ldots, \bF_{K(K-1)}(\bbU,\bbV)\right) = 0, \nonumber
\end{align}
for all $\{\bar\bU_k,\; \bar\bV_k\}_{k=1}^K$.
Notice that the polynomial $p$ is independent of the channel matrices $\left\{\bar\bH^{(1)}_{kj}\right\}_{k\neq j}$, even though it does depend on the matrices $\left\{ \bar\bH_{kj}^{(2)},\bar\bH_{kj}^{(3)},\bar\bH_{kj}^{(4)}\right\}_{k\neq j}$. {\color{black} When viewed as a polynomial of the matrix variable $\bX:=\left(\bbH^{(1)}_{12},\bbH^{(1)}_{13},\ldots, \bbH^{(1)}_{K(K-1)}\right)$, $p(\cdot)$ can be expanded locally at $\bar\bX:=(\bF_{12}(\bbU,\bbV),\bF_{13}(\bbU,\bbV),\ldots,\bF_{K(K-1)}(\bbU,\bbV))$ using \eqref{eq:taylor}:}
\begin{eqnarray}\nonumber
&&p\left(\bbH^{(1)}_{12},\bbH^{(1)}_{13},\ldots, \bbH^{(1)}_{K(K-1)}\right) \\
&&\quad= p\left(\bF_{12}(\bbU,\bbV),\bF_{13}(\bbU,\bbV),\ldots,\bF_{K(K-1)}(\bbU,\bbV)\right)\nonumber\\
&&\quad\ \  + \sum_{k\neq j}\tr\left((\bbH_{kj}^{(1)}-\bF_{kj}(\bbU,\bbV))\bQ_{kj}(\bbU,\bbV)\right),\nonumber%\label{eq4}
\end{eqnarray}
for all  $\{\bbU_k,\;\bbV_k\}_{k=1}^K$, where $\bQ_{kj}$ is some polynomial matrix of size $d_j\times d_k$. % with entries corresponding to, among other things, the partial derivatives (of various order) of $p$ and different powers of $(\bbH_{kj}^{(1)} - \bF_{kj})$'s\footnote{{\color{black}In the Taylor expansion of polynomial~$f(\cdot)$ around the point~$x^*$, we have
%\[
%f(x) = f(x^*) + \sum_{i=1}^m \frac{(x-x^*)^i}{i!} f^{(i)}(x^*),
%\] where~$m$ is the polynomial order of~$f(\cdot)$. Therefore, by factoring out the term~$x-x^*$, we can write~$f(x) = f(x^*) + (x-x^*) q(x)$.}}. %Thus, the entries of $\bQ_{kj}$ are polynomial functions of the entries in the matrices $\{\bbU_k,\bbV_k\}_{k=1}^K$.}
Combining the above two identities yields
\begin{eqnarray}\nonumber
\!\!\!\!\!\!\!\!\!&&p\left(\bbH^{(1)}_{12},\bbH^{(1)}_{13},\ldots, \bbH^{(1)}_{K(K-1)}\right)\\
\!\!\!\!\!\!\!\!\!&&\quad= \sum_{k\neq j}\tr\left((\bbH_{kj}^{(1)}-\bF_{kj}(\bbU,\bbV))\bQ_{kj}(\bbU,\bbV)\right).\label{eq4}
\end{eqnarray}
Notice that this equality holds for all choices of $\{\bbU_k,\;\bbV_k\}_{k=1}^K$.
If the interference alignment condition \eqref{eq2} holds, then we have
\[
\bbH_{kj}^{(1)}-\bF_{kj}(\bbU,\bbV)=0, \quad \mbox{for all $k,\;j$ with $k\neq j$},
\]
for some special choices of the matrices $\{\bbU_k,\bbV_k\}_{k=1}^K$.
 Substituting this condition into the right hand side of \eqref{eq4}, we obtain
\begin{equation}\label{eq5}
p\left(\bbH^{(1)}_{12},\bbH^{(1)}_{13},\ldots, \bbH^{(1)}_{K(K-1)}\right)=0.
\end{equation}

Notice that the polynomial $p$ is independent of the channel matrices $\{\bbH^{(1)}_{kj}\}_{k\neq j}$. Under our channel model, the channel matrices $\{\bbH^{(1)}_{kj}\}_{k\neq j}$ are drawn from a continuous probability distribution. It follows that the condition \eqref{eq5} cannot hold unless $p$ is identically zero, which contradicts the requirement $p\neq 0$.
\end{proof}
%\begin{rmk}

% and independent of {\color{black}$\{\bbH^{(2)}_{kj},\bbH^{(3)}_{kj},\bbH^{(4)}_{kj}\}_{k \neq j}$}.
%\end{rmk}
Theorem~\ref{thm:bd} settles the conjecture of \cite{Jafarbound} in one direction, namely, the improperness of polynomial system \eqref{IA1} and \eqref{IA2} implies the infeasibility of interference alignment. From the proof of Theorem~\ref{thm:bd}, it can be seen that the upper bound \eqref{IAbound} holds for any choice of fixed channel matrices~{\color{black}$\{\bbH^{(2)}_{kj},\bbH^{(3)}_{kj},\bbH^{(4)}_{kj}\}_{k \neq j}$} as long as {\color{black} the channel matrices $\{\bbH^{(1)}_{kj}\}_{k \neq j}$ are generic.

Also, we remark that
the proof technique for Theorem~\ref{thm:bd} can be used to bound the DoF for a single antenna parallel interference channel (e.g., the OFDM channel). In particular, consider a single input single output interference channel with $M$~channel extensions, i.e., the channel matrices are diagonal and of the size~$M\times M$. Assuming each user transmits one data stream ($d_k=1$ for all $k$), we can check that the properness of the interference alignment condition \eqref{IA1}-\eqref{IA2} is equivalent to $K+1\le 2M$ (see \cite[Theorem 1]{Jafarbound}). Using a completely identical proof, we can show that the properness condition $K+1\le 2M$ is a necessary condition for the feasibility of interference alignment. This implies that for the single beam case the total DoF per channel extension is upper bounded by 2, regardless of the number of channel extensions. This DoF bound has also been proposed recently in \cite{Berry}.}

\subsection{The Converse Direction}

In the remainder of this section, we consider the converse of Theorem~\ref{thm:bd}. In particular, we show that the upper bound in Theorem~\ref{thm:bd} is tight for a special case where all users have the same DoF~$d$ and number of antennas is divisible by~$d$. In this case, we have $K(K-1)$ matrix equations in~\eqref{eq2}, each giving rise to~$d^2$ scalar equations. For any subset of these matrix equations indexed by $\mathcal{I}$, with $\mathcal{I} \subseteq \mathcal{J}$,
the number of corresponding scalar equations is equal to $d^2|{\cal I}|$, whereas the number of scalar variables involved in the equations indexed by ${\cal I}$ is
$$\left(\sum_{k:(k,j)\in {\cal I}}(M_k-d)+\sum_{j:(k,j)\in {\cal I}}(N_j-d)\right)d.$$\\

 %provides an upper bound for achievable degrees of freedom in the system.
The next result shows that the bound in Theorem~\ref{thm:bd} is tight if the polynomial system \eqref{eq2} defining interference alignment is proper, i.e., for each ${\cal I}\subseteq {\cal J}$, the number of variables involved in each set of equations indexed by ${\cal I}$ is no less than $d^2|{\cal I}|$, the number of scalar equations. {\color{black} The proof of this result uses the Implicit Function Theorem which involves checking the Jacobian matrix of the polynomial map \eqref{eq:F} is nonsingular at some channel realization $\{\bar\bH_{kj}\}_{k\neq j}$. Notice that the feasibility of interference alignment condition \eqref{eq2} at a given channel realization  $\{\bar\bH_{kj}\}_{k\neq j}$ is equivalent to  $\{\bar\bH^{(1)}_{kj}\}_{k\neq j}$ being contained in the image of the polynomial map \eqref{eq:F} which is defined by $\{\bar\bH^{(2)}_{kj},\bar\bH^{(3)}_{kj},\bar\bH^{(4)}_{kj}\}_{k\neq j}$. Fix
a generic choice of $\{\bar\bH^{(2)}_{kj},\bar\bH^{(3)}_{kj},\bar\bH^{(4)}_{kj}\}_{k\neq j}$ for which the Jacobian of the polynomial map \eqref{eq:F} is nonsingular. The Implicit Function Theorem allows us to establish the existence of a locally invertible map from the space of channel submatrices $\{\bar\bH^{(1)}_{kj}\}_{k\neq j}$ to the space of beamforming matrices, and that the image of this polynomial map \eqref{eq:F} is locally full-dimensional.  Therefore, for all channel submatrices near the given channel realization $\{\bar\bH^{(1)}_{kj}\}_{k\neq j}$, the interference alignment condition \eqref{eq2} can be satisfied by some beamforming matrices. By Chevalley's Theorem from algebraic geometry \cite{Matsumura} (see also the discussion at the end of Section~\ref{sec:preliminary}), the ``local full-dimensionality" of the image of \eqref{eq:F} implies that this image, which is a constructible set, must contain a nonempty Zariski open set. As a result, the whole image of polynomial map \eqref{eq:F} contains all generically generated channel sub-matrices $\{\bar\bH^{(1)}_{kj}\}_{k\neq j}$. Since the choice of channel submatrices $\{\bar\bH^{(2)}_{kj},\bar\bH^{(3)}_{kj},\bar\bH^{(4)}_{kj}\}_{k\neq j}$ is also generic, this then establishes the feasibility of interference alignment for all generically generated channel matrices $\{\bar\bH_{kj}\}_{k\neq j}$.}

\begin{thm}
\label{thm:achievability}
Assume that all users have the same DoF~$d_k = d$, where~$1 \leq d \leq \min \{M_k,N_k\},\;\forall k$. Furthermore, suppose that $M_k$ and $N_k$ are divisible by $d$ for all~$k$. Then interference alignment is achievable for generic channel coefficients if and only if for each subset $\cal I$ of equations in \eqref{eq2}, the number of variables involved in these equations is no less than the number of matrix equations times $d^2$, or equivalently,
\begin{equation}\label{assumption}
|{\cal I}|d\le \sum_{k:(k,j)\in {\cal I}}(M_k-d)+\sum_{j:(k,j)\in {\cal I}}(N_j-d), \quad \quad \forall \; \mathcal{I} \mbox{ with }\mathcal{I}\subseteq \mathcal{J}.
\end{equation}
\end{thm}
\begin{proof}
First of all, the ``only if" direction is a direct consequence of Theorem~\ref{thm:bd}.
We now focus on the ``if" direction. Consider the polynomial map that we get by concatenating all maps in~\eqref{eq:F} for all $(k,j) \in \mathcal{J}$, i.e.,
\begin{equation}\label{eq:Fs}
\begin{split}
\bF_{12}(\bbU,\bbV) &= -\left(\bar\bU_1^H \bar\bH_{12}^{(3)}  + \bar\bH_{12}^{(2)}\bar\bV_2 + \bar\bU_1^H \bar\bH_{12}^{(4)} \bar\bV_2\right),\\
\bF_{13}(\bbU,\bbV) &= -\left(\bar\bU_1^H \bar\bH_{13}^{(3)}  + \bar\bH_{13}^{(2)}\bar\bV_3 + \bar\bU_1^H \bar\bH_{13}^{(4)} \bar\bV_3\right),\\
&\ \ \vdots\\
\bF_{K(K-1)}(\bbU,\bbV) &= -\left(\bar\bU_{K}^H \bar\bH_{K(K-1)}^{(3)}  + \bar\bH_{K(K-1)}^{(2)}\bar\bV_{K-1} + \bar\bU_{K}^H \bar\bH_{K(K-1)}^{(4)} \bar\bV_{K-1}\right),
\end{split}
\end{equation}
which maps the variables~$\{\bbU_k,\bbV_k\}_{k=1}^K$ to the $\{\bF_{k,j}\}_{k\neq j}$ space.
We will first show that for a specific set of channel matrices, the rank of the Jacobian of this polynomial map is~$K(K-1)d^2$, equal to the number of equations. Hence, if we restrict the equations to a subset of variables of size~$K(K-1)d^2$, the determinant of the Jacobian matrix of the polynomial map \eqref{eq:Fs} does not vanish identically. {\color{black} This step will establish the existence of a locally invertible map from the space of beamforming matrices to the space of channel matrices. By Chevalley's Theorem (see \cite[Chapter 2, 6.E.]{Matsumura}), this image is a constructible subset under Zariski topology. This, plus the fact that the image is locally full-dimensional, implies that the interference alignment condition \eqref{eq2} is feasible for all generically chosen channel matrices. This then will show the ``if" direction of Theorem~\ref{thm:achievability}.}

{\color{black} To show the nonsingularity of the Jacobian matrix, we need to remove some redundant variables in $\{\bar\bU_k,\bar\bV_j\}_{k, j}$ (this occurs when there are more variables than equations), and then construct a specific set of channel matrices $\{\bH_{kj}^{(2)},\bH_{kj}^{(3)},\bH_{kj}^{(4)}\}_{k\neq j}$ and a solution $\{\bar\bU_k,\bar\bV_j\}_{k, j}$ at which the Jacobian matrix of \eqref{eq:Fs} is nonsingular. Before providing a rigorous description for such a construction, we first consider a toy example with $K=3$ users where $M_k=3, N_k=2, d_k=1,$ for $k=1,2,3$. For this specific example, the assumption \eqref{assumption} is satisfied and the equations in~\eqref{eq:Fs} can be rewritten as
\begin{align}
\bF_{12}(\bbU,\bbV) &= -\left(\bar\bU_1^H \bar\bH_{12}^{(3)}  + \bar\bH_{12}^{(2)}\bar\bV_2 + \bar\bU_1^H \bar\bH_{12}^{(4)} \bar\bV_2\right),\nonumber\\
\bF_{13}(\bbU,\bbV) &= -\left(\bar\bU_1^H \bar\bH_{13}^{(3)}  + \bar\bH_{13}^{(2)}\bar\bV_3 + \bar\bU_1^H \bar\bH_{13}^{(4)} \bar\bV_3\right),\nonumber\\
\bF_{21}(\bbU,\bbV) &= -\left(\bar\bU_2^H \bar\bH_{21}^{(3)}  + \bar\bH_{21}^{(2)}\bar\bV_1 + \bar\bU_2^H \bar\bH_{21}^{(4)} \bar\bV_1\right),\nonumber\\
\bF_{23}(\bbU,\bbV) &= -\left(\bar\bU_2^H \bar\bH_{23}^{(3)}  + \bar\bH_{23}^{(2)}\bar\bV_3 + \bar\bU_2^H \bar\bH_{23}^{(4)} \bar\bV_3\right),\nonumber\\
\bF_{31}(\bbU,\bbV) &= -\left(\bar\bU_3^H \bar\bH_{31}^{(3)}  + \bar\bH_{31}^{(2)}\bar\bV_1 + \bar\bU_3^H \bar\bH_{31}^{(4)} \bar\bV_1\right),\nonumber\\
\bF_{32}(\bbU,\bbV) &= -\left(\bar\bU_3^H \bar\bH_{32}^{(3)}  + \bar\bH_{32}^{(2)}\bar\bV_2 + \bar\bU_3^H \bar\bH_{32}^{(4)} \bar\bV_2\right),\nonumber
\end{align}
where $\bbV_k = [v_{k_1}\; v_{k_2}]^T \in \mathbb{C}^{2 \times 1}$, $\bbU_k = [u_k] \in \mathbb{C}$,  for $k=1,2,3$, and $\bbH_{kj}^{(2)} = [\barh_{kj}^{(2),1} \; \barh_{kj}^{(2),2}]^T \in \mathbb{C}^{2 \times 1}$, $\bbH_{kj}^{(3)} = [\barh_{kj}^{(3)}] \in \mathbb{C}$, for $k\neq j$. If we set $\bbH_{kj}^{(4)} = 0$ for all channels, one can write the Jacobian of $[\bF_{12}\;\bF_{13}\;\bF_{21}\;\bF_{23}\;\bF_{31}\;\bF_{32}]$ with respect to the variables~$[u_1\; u_2\; u_3\; v_{1_1}\; v_{1_2}\; v_{2_1}\; v_{2_2}\; v_{3_1}\; v_{3_2}]$ as
\begin{align}
\left[
\begin{array}{cccccc}
-\barh_{12}^{(3)} & -\barh_{13}^{(3)} & 0 & 0 & 0 & 0\\
0 & 0 & -\barh_{21}^{(3)} & -\barh_{23}^{(3)} & 0 & 0\\
0 & 0 & 0 & 0 & -\barh_{31}^{(3)} & -\barh_{32}^{(3)}\\
0 & 0 & -\barh_{21}^{(2),1} & 0 & -\barh_{31}^{(2),1} & 0\\
0 & 0 & -\barh_{21}^{(2),2} & 0 & -\barh_{31}^{(2),2} & 0\\
-\barh_{12}^{(2),1} & 0 & 0 & 0 & 0 & -\barh_{32}^{(2),1}\\
-\barh_{12}^{(2),2} & 0 & 0 & 0 & 0 & -\barh_{32}^{(2),2}\\
0 & -\barh_{13}^{(2),1} & 0 & -\barh_{23}^{(2),1} & 0 & 0\\
0 & -\barh_{13}^{(2),2} & 0 & -\barh_{23}^{(2),2} & 0 & 0\\
\end{array}
\right]. \nonumber
\end{align}
One can easily observe that by removing the variables~$\{v_{1_1},v_{2_1},v_{3_2}\}$ and setting
\begin{align}
\barh_{12}^{(3)} = \barh_{23}^{(3)} = \barh_{31}^{(3)} = \barh_{13}^{(2),1} = \barh_{21}^{(2),2} = \barh_{32}^{(2),2} = 1, \nonumber \\
\barh_{13}^{(3)} = \barh_{21}^{(3)} = \barh_{32}^{(3)} = \barh_{12}^{(2),2} = \barh_{31}^{(2),2} = \barh_{23}^{(2),1} = 0,\nonumber
\end{align}
the Jacobian of the mapping~\eqref{eq:Fs} with respect to the remaining variables becomes
\begin{align}
\left[
\begin{array}{cccccc}
-1 & 0 & 0 & 0 & 0 & 0\\
0 & 0 & 0 & -1 & 0 & 0\\
0 & 0 & 0 & 0 & -1 & 0\\
0 & 0 & -1 & 0 & 0 & 0\\
0 & 0 & 0 & 0 & 0 & -1\\
0 & -1 & 0 & 0 & 0 & 0\\
\end{array}
\right], \nonumber
\end{align}
which is clearly nonsingular since there exists exactly one nonzero element in each column/row.

Next we argue that the above construction procedure can be generalized to the case where $M_k$ and $N_k$ are divisible by~$d$, provided that the assumption~\eqref{assumption} is satisfied. The construction of these channel/beamforming matrices and the removal of redundant variables are outlined below.
First, we set $\bH_{kj}^{(4)}=\bzero$, for all $k\neq j$. Then we choose $\{\bar\bU_k,\bar\bV_j\}_{k,j}$ arbitrarily.
%The channel matrices $\bH_{kj}^{(1)}$ will be chosen last according to
%\[
%\bar\bH_{kj}^{(1)}=-
%\bar\bU^H_k\bar\bH_{kj}^{(3)}-\bar\bH_{kj}^{(2)}\bar\bV_j-\bar\bU^H_k\bH_{kj}^{(4)}\bar
%\bV_j,\quad k\neq j.
%\]
%With this choice of channel matrices $\{\bar\bH_{kj}^{(1)}\}$, the interference alignment condition \eqref{eq2} always holds for any choice of channel matrices $\{\bar\bH_{kj}^{(3)},\bar\bH_{kj}^{(2)}\}_{k\neq j}$.
It remains to specify $\{\bar\bH_{kj}^{(3)},\bar\bH_{kj}^{(2)}\}_{k\neq j}$. We should do so to ensure that the corresponding Jacobian matrix of \eqref{eq:Fs} at $\{\bar\bU_k,\bar\bV_j\}_{k,j}$ is nonsingular. Since $M_k$ and $N_k$ are divisible by~$d$, we can partition our variables into blocks of size~$d\times d$ and rewrite the mapping~\eqref{eq:Fs} as
\begin{equation}
\label{blockIA}
\bF_{kj}(\bbU,\bbV)=
- \left[  \bbU_{k_1}^H\;\bbU_{k_2}^H\ldots \bbU_{k_{s_k}}^H\right]\left[
\begin{array}{c}
\bbH_{kj}^{(3),1}\\
\bbH_{kj}^{(3),2}  \\
\vdots \\
\bbH_{kj}^{(3),s_k}  \\
\end{array}
\right] - \left[  \bbH_{kj}^{(2),1}\;\bbH_{kj}^{(2),2}\ldots \bbH_{kj}^{(2),t_j}\right]\left[
\begin{array}{c}
\bbV_{j_1}\\
\bbV_{j_2}  \\
\vdots \\
\bbV_{j_{t_j}}  \\
\end{array}
\right], \quad \forall \; k \neq j,
\end{equation}
where $s_k = \frac{M_k}{d}-1, t_j = \frac{N_j}{d}-1$, and $\bbU_{k_i}, \bbV_{j_\ell}, \bbH_{kj}^{(2),i}, \bbH_{kj}^{(3),\ell} \in \mathbb{C}^{d\times d}$.
Consider a bipartite graph~$G$ where the vertices are partitioned into two sets~$\mathcal{X}$ and $\mathcal{Y}$.
Each block of variables will correspond to a node in~$\mathcal{X}$, while each matrix equation in \eqref{blockIA} will correspond to a node in~$\mathcal{Y}$. We draw an edge between a node $x\in \mathcal{X}$ and a node $y \in \mathcal{Y}$ if the block of variables corresponding to node~$x$ appears in the equation corresponding to node~$y$. When viewed on the bipartite graph $G$, the assumption~\eqref{assumption} simply says that for any given set of nodes~$\mathcal{S} \subseteq \mathcal{Y}$, the cardinality of the neighbors of~$\mathcal{S}$ in $\mathcal{X}$ is no smaller than the cardinality of $\mathcal{S}$. This condition is precisely what is required to ensure the existence of a complete matching in $G$ covering all nodes in $\mathcal{Y}$ (Hall's theorem, see \cite[Theorem 3.1.11]{West}). Now consider a fixed complete matching in $G$. Let $A \subseteq \mathcal{X}$ be the set of vertices that are not matched to a node in $\mathcal{Y}$. Then, we can set
to zero all the blocks of the variables corresponding to the vertices in~$A$, i.e., we can remove them from our equations.
Now we choose the rest of the channel matrices so that the determinant of the Jacobian with respect to the remaining variables is nonzero. To this end, we set $\bbH_{kj}^{(3),p} = 0$ if the node for $\bbU_{k_p}$ is not matched to the node in $\mathcal{Y}$ corresponding to the equation $\bF_{kj}$.
Similarly, we set $\bbH_{kj}^{(2),q} = 0$ if $\bbV_{j_q}$ is not matched to $\bF_{kj}$. Moreover, we set all the remaining channel sub-matrices to the $d\times d$ identity matrix. Since this construction is based on a complete matching, it is not hard to see
that the Jacobian for the whole system is a block permutation matrix, with nonzero blocks equal to the negative $d\times d$ identity matrix. Hence the determinant of the Jacobian matrix is equal to the product of
the determinant of all nonzero blocks (up to sign), which is clearly nonzero in our case.}
This completes the description of the procedure to remove potential redundant variables, as well as the procedure to construct all the channel matrices $\{\bH^{(2)}_{kj},\bH^{(3)}_{kj},\bH^{(4)}_{kj}\}_{k\neq j}$ and the beamforming solution $\{\bar\bU_k,\bar\bV_j\}_{k,j}$. The Jacobian matrix of \eqref{eq:Fs} is nonsingular at this constructed channel realization and beamforming solution.
%Since nonsingularity of the Jacobian is a generic property, this further implies that the Jacobian of \eqref{eq:Fs} is nonsingular generically for all channel realizations.
{\color{black} Figure~$\ref{Fig:ToyExample}$ illustrates the construction of graph~${G}$ and a complete matching (in solid lines) for the aforementioned toy example.}\\

\begin{figure}[htbp]
\centering
\includegraphics[width=2.5in]{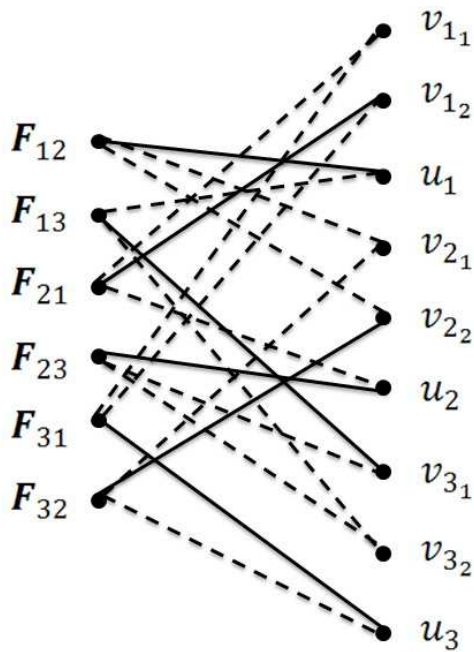}
\caption{The bipartite graph~${G}$ and a complete matching for the toy example}
\label{Fig:ToyExample}
\end{figure}

To complete the proof, we fix a generic choice of $\{\bbH^{(2)}_{kj},\bbH^{(3)}_{kj},\bbH^{(4)}_{kj}\}_{k\neq j}$ for which the Jacobian of \eqref{eq:Fs} is nonsingular. Let~$n$ be the total number of remaining scalar variables in $\{\bar\bU_k,\bar\bV_j\}_{k,j}$ after removing the redundant variables. Notice that $n$ is the same as the number of scalar equations, i.e., $n = d^2 K (K-1)$. Let $R_1 = \mathbb{C}[h_1,h_2,\ldots,h_n]$ and $R_2 = \mathbb{C} [\bar{u}_1,\bar{u}_2,\ldots, \bar{u}_m, \bar{v}_1, \ldots,  \bar{v}_{n-m}]$ be two polynomial rings where $\bar{u}_i$'s and $\bar{v}_j$'s are the entries of the matrices $\{\bbU_k\}_{k=1}^K$ and $\{\bbV_k\}_{k=1}^K$ (after removing the redundant variables), and $h_1, h_2, \ldots, h_n$ are the entries of the matrices~$\{\bbH^{(1)}_{kj}\}_{k\neq j}$. Consider $\{f_i\}_{i=1}^n$ (the components of $\mathbf{F}_{kj}$'s in \eqref{eq:Fs}) as the functions of $\bar{u}$'s and $\bar{v}$'s, i.e., $f_i$'s are polynomials in $R_2$.
%Then, {\color{black}the polynomials $\{f_i\}_{i=1}^n$} define a ring homomorphism $\phi: R_1 \mapsto R_2$. This map defines another
These polynomials define a map~$a_\phi$ which maps a point $c = (c_1,c_2,\ldots,c_n)$  to $(f_1(c),f_2(c),\ldots,f_n(c))$.
%Consider the map $a_\phi: A_{\mathbb{C},2}^n \mapsto A_{\mathbb{C},1}^n$ defined by $\{f_i\}_{i=1}^n$ where $A_{\mathbb{C},1}^n$ and $A_{\mathbb{C},2}^n$ are the corresponding affine spaces of $R_1$ and $R_2$, respectively. Namely, a point $c = (c_1,c_2,\ldots,c_n)$ is mapped to $(f_1(c),f_2(c),\ldots,f_n(c))$ through the map~$a_\phi$.
According to the Chevalley Theorem (see \cite[Chapter 2, 6.E.]{Matsumura}), the image of this map is a Zariski constructible subset of~$A_{\mathbb{C},1}^n$, where $A_{\mathbb{C},1}^n$ is the corresponding affine space of~$R_1$. Since the Jacobian of the set~$\{f_1,f_2,\ldots,f_n\}$ with respect to the variables~$\{\bar{u}_1,\bar{u}_2,\ldots, \bar{u}_m, \bar{v}_1, \ldots,  \bar{v}_{n-m}\}$ is nonsingular generically for all channel realizations, it follows from the Implicit Function Theorem that the dimension of the image of $a_\phi$ is~$n$. %To do so, we need to find a channel realization for which the Jacobian of the set~$\{f_1,f_2,\ldots,f_n\}$ with respect to the variables~$\{\bar{u}_1,\bar{u}_2,\ldots, \bar{u}_m, \bar{v}_1, \ldots,  \bar{v}_{n-m}\}$ is not zero.
Note that the image of~$a_{\phi}$ is a Zariski constructible subset of~$A_{\mathbb{C},1}^n$ (see Chevalley Theorem \cite{Matsumura}, Ch.\ 2, 6.E.) and it has full dimension. Hence, the image contains a Zariski open subset of~$A_{\mathbb{C},1}^n$ (see the discussion in section~\ref{sec:preliminary}). Let~$\mathcal{U}$ be that Zariski open subset of~$A_{\mathbb{C},1}^n$ in the image. Since~$\mathcal{U}$ is in the image of the map~$a_\phi$, there exists a solution for interference alignment equations for any choice of~$\{\bbH^{(1)}_{kj}\}_{k\neq j}$ in~$\mathcal{U}$,
%Moreover, the complement of the set~$\mathcal{D}$ is a Zariski closed set, i.e., there exists a set of finite polynomials that vanishes on the complement of~$\mathcal{D}$. Hence, the complement of~$\mathcal{D}$ occurs with probability zero
which implies that interference alignment is feasible for generic choice of $\{\bbH^{(1)}_{kj}\}_{k\neq j}$.
Since the choice of channel matrices $\{\bbH^{(2)}_{kj},\bbH^{(3)}_{kj},\bbH^{(4)}_{kj}\}_{k\neq j}$ is also generic, this completes the proof of the ``if" direction.
\end{proof}
\bigskip

Notice that the condition \eqref{assumption} is equivalent to the properness of the polynomial system \eqref{eq2} defining interference alignment. For symmetric systems with $M_k=M$, $N_k=N$ for all $k$, this condition simplifies to $M+N\ge d(K+1)$ (see \cite[Theorem 1]{Jafarbound}). Thus, each user can achieve $d$ degrees of freedom as long as $M+N\ge d(K+1)$ and that $d$ divides both $M$ and $N$. {\color{black} In a concurrent work, the authors of~\cite{Berkeley-IA} obtained a similar result under a different set of assumptions. More specifically, they considered the symmetric case with~$M_k = N_k =M,\; d_k = d$ for all $k$, and proved that the feasibility of interference alignment in this case is equivalent to $2M\ge d(K+1)$. This result and Theorem~\ref{thm:achievability} are complementary to each other. In particular, Theorem~\ref{thm:achievability} is applicable to non-symmetric systems, but does require an extra condition about the divisibility of the number of antennas by the number of data streams. When $K$ is odd and $(K+1)d = 2M$, then $M$ must be divisible by~$d$. This case is then covered by both Theorem~\ref{thm:achievability} and the result in \cite{Berkeley-IA}. However, for the case where $K$ is even and $(K+1)d \leq 2M$, Theorem~\ref{thm:achievability} is no longer applicable, whereas \cite{Berkeley-IA} shows that the interference alignment is achievable.}

A few other remarks are in order.
\begin{enumerate}
\item
Reference \cite{Jafarbound} also considered the case $d_k=1$ and used the Bernshtein's theorem to numerically compute the number of solutions, and therefore prove the feasibility, for the resulting polynomial system \eqref{IA1}-\eqref{IA2} when the number of antennas are small. In contrast, Theorem~\ref{thm:achievability} shows the feasibility of single beam interference alignment for all values of $M_k$, $N_k$ as long as the system is proper.
%\end{rmk}
%\begin{rmk}
%The reference \cite{IANPhard} shows that, given a set of channel matrices, checking the feasibility of the interference alignment conditions \eqref{IA1}-\eqref{IA2} is NP-hard when $M_k> 2$, $N_k>2$. Under the setting of Theorem~\ref{thm:achievability}, the interference alignment fails only for a measure zero set of channels. However, the results of \cite{IANPhard} still imply that, even if the interference alignment is feasible, finding the actual linear transmit/receive beamformers to achieve it is still a NP-hard problem when the number of users is large.
%\end{rmk}

\item %\begin{rmk}
As shown in Theorem~\ref{thm:bd}, the condition \eqref{I2} is necessary. For example, the system~$K=2, M=N=3, d=2$ satisfies {\color{black} the inequality \eqref{IAbound}. However, the system of equations~\eqref{IA1}-\eqref{IA2} is infeasible for generic choice of channel coefficients. This further shows that the \textit{properness} property in~\cite{Jafarbound} does not imply feasibility in general, a fact that was first pointed out in \cite[example 17]{Jafarbound}}.
%\end{rmk}
%\begin{rmk}
\item Theorem~\ref{thm:achievability} does not contradict the NP-hardness result of \cite{IANPhard}. Given a set of channel matrices, checking the feasibility of the interference alignment conditions \eqref{IA1}-\eqref{IA2}  when $M_k\ge3$ and $N_k\ge3$,  is NP-hard. It is true that, under the setting of Theorem~\ref{thm:achievability}, the interference alignment fails only for a measure zero set of channels. However, for systems not satisfying the conditions of Theorem~\ref{thm:achievability}, checking the feasibility of interference alignment can be hard. Moreover, the results of \cite{IANPhard} imply that, even if a given tuple of DoF is known to be achievable via interference alignment, finding the actual linear transmit/receive beamformers to achieve it is still a NP-hard problem when the number of users is large.
%If $d =1$, Theorem~\ref{thm:achievability} states that interference alignment is feasible if $\{M_k,N_k\}_{k=1}^K$ satisfies the condition~\eqref{assumption}.
\item The condition~\eqref{assumption} implies the condition~\eqref{I2} {\color{black} if the number of antennas at each transceiver is divisible by~$d$}. In fact, by choosing $\mathcal{I} = \{(k,j)\}$, condition~\eqref{assumption} implies that $d \leq M_k +N_j -2d$ and hence the condition~$\eqref{I2}$ is satisfied.
{\color{black}
\item Theorem~\ref{thm:achievability} assumes that both $M_k$ and $N_k$ are divisible by $d$. This condition can be weakened for a symmetric system where $M_k = M,\; N_k = N,\; d_k =d$, for all $k$. In particular, assume that only $M$ (not $N$) is divisible by~$d$ and $M,N \geq d$. If the properness condition $(K+1)d \leq M+N$ holds, then we can construct a reduced MIMO interference channel with~$N' = M - d(K+1)$ receive antennas for each user, where $M+N' = d(K+1)$ and $M, N'$ are divisible by~$d$. By Theorem~\ref{thm:achievability}, the interference alignment condition for the reduced interference channel is feasible and therefore, so is the interference alignment condition for the original channel since the latter has more antennas. This shows that if $M$ is divisible by~$d$ and $M,N \geq d$, then the interference alignment system \eqref{IA1}--\eqref{IA2} is feasible for generic choice of channel coefficients if and only if~$(K+1)d \leq M+N$. By symmetry, the same conclusion holds for the case where $N$ is divisible by~$d$.
}
\end{enumerate}

\section{Simulation Results}

{\color{black}In this section, we use the theoretical DoF upper bounds to benchmark an existing algorithm for sum-rate maximization. We generate MIMO interference channels using the standard Rayleigh fading model. The numerical experiments are averaged over 100 Monte Carlo runs.}

%[[[[Replot the figures after normalizing by log SNR at high SNR. In other words, plot DoF vs M+N]]]
%
%[[[Specify which algorithm (WMMSE?) is used and how it works. Cite reference.]]]
%
%[[[When M+N changes, the total achievable DoF changes theoretically (Theorem 1), and empirically (WMMSE). Plot the two curves on the same figure, and compare.]]]

%In the first numerical experiment, we consider a MIMO interference channel with $3$~users. We further consider two different cases of $M=N=2, d=1$ and $M=N=4, d=2$. In both cases, the bounds in Theorem~\ref{thm:bd} and Theorem~\ref{thm:achievability} are satisfied, suggesting the interference alignment is achievable. We plot the sum-rate results of the Distributed Interference Alignment (DIA) algorithm \cite{Jafar2} (see also \cite{heath}). The ``Predicted Slope" line has the slope corresponding to the DoF in the system and it starts from the same point as the DIA method. The closeness of the two curves suggests that the expected interference alignment has been achieved. \\%As we can see, the two curves are close to each other which approves the achievability of spatial interference alignment.\\
%\begin{figure}[htbp]
%\centering
%\includegraphics[width=3.5in]{K3M4-2d2-1.eps}
%\caption{Sum-rate versus SNR(dB)}
%\label{fig:K3M4-2d2-1}
%\end{figure}
%In the second numerical experiment,
We consider a MIMO interference channel where each transmitter/receiver is equipped with~$3$ antennas. For different number of users in the system, we maximize the sum-rate using the WMMSE algorithm \cite{WMMSE} at increasingly high SNRs. We estimate the slope of the sum-rate versus SNR and use it to approximate the achievable total DoF. We then compare it with the value of theoretical upper bound given by the conditions in Theorem~\ref{thm:bd}. The maximum gap of the two curves is one, but it is not clear if the gap is due to the weakness of the WMMSE algorithm or the DoF upper bound.
\begin{figure}[htbp]
\centering
\includegraphics[width=3.5in]{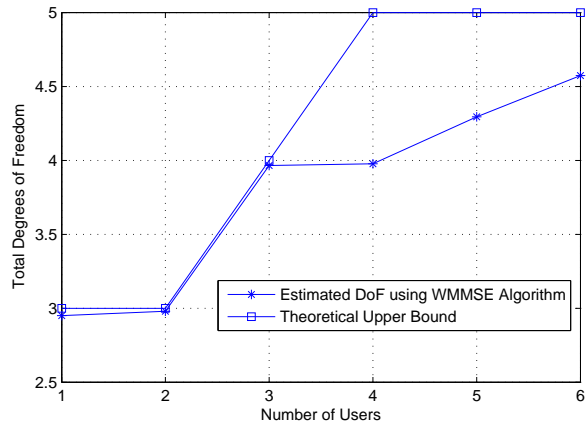}
\caption{Achievable DoF and theoretical upper bound}
\label{SlopeK2}
\end{figure}
%In the second numerical experiment, we explore systems where the bound in Theorem~\ref{thm:bd} and Theorem~\ref{thm:achievability} is not satisfied. The first plot corresponds to the system with $K=3$ users, $M=N=4$ antennas at the transmitter and receiver side, and we allocate $d=3$ DoF to each user. As we expected, the spatial interference alignment is not possible in these cases and hence the resulted plot is not linear (unlike the plots in the previous simulation).
%\begin{figure}[htbp]
%\centering
%\includegraphics[width=3.5in]{K3M4-2NoIA.eps}
%\caption{Sum rate versus SNR(dB) when interference alignment is infeasible}
%\label{fig:K3M4-2NoIA}
%\end{figure}
%In the last numerical result, we utilize the WMMSE algorithm \cite{WMMSE} to study the behavior of achievable sum rate versus the number of antennas. In this simulation, we consider~$K=4$ user interference channel where each transmitter/receiver~$k$ is equipped with $M_k=N_k = M$ antennas. We fix the power budget of each user to~$30 dB$ and plot the sum rate versus the number of antennas. From this plot, we can see that the total DoF in the system increases as the number of antennas increases, even in the case of nonsymmetric channel where~$d_i$'s could not be the same.
%\begin{figure}[htbp]
%\centering
%\includegraphics[width=3.5in]{K4M1-6RatevsM.eps}
%\caption{Sum rate versus the number of antennas}
%\label{fig:K4M1-6RatevsM}
%\end{figure}

\bigskip
\noindent
\textbf{Acknowledgement:} We are grateful to the authors of \cite{Berkeley-IA} for sharing their concurrent work.


\begin{thebibliography} {TL99}
\bibitem {Jafar1}  V. Cadambe and S. Jafar, ``Interference Alignment and the Degrees
of Freedom of the $K$ User Interference Channel," \textit{IEEE Trans. on Information Theory}, Vol. 54, No.\ 8, August 2008.
\bibitem{Jafarbound} C.M.\ Yetis, T.\ Gou, S.A.\ Jafar, and A.H.\ Kayran
``On Feasibility of Interference Alignment in MIMO Interference Networks", \textit{IEEE Trans. on Signal Processing,} Vol.\ 58, pp.\ 4771-4782, 2010.

%\bibitem {Jafar2}  K.\ Gomadam, V.R.\ Cadambe, and S.A.\ Jafar ``Approaching the Capacity of Wireless Networks through Distributed Interference Alignment," \textit{IEEE GLOBECOM}, 2008.

\bibitem{waterloo} M. Maddah-Ali, A. Motahari, and A. Khandani, ``Communication over MIMO
X Channels: Interference Alignment, Decomposition, and Performance Analysis," \emph{IEEE Trans.\ on Information Theory,} Vol.\ 54, pp.\ 3457--3470, August 2008.

\bibitem{JafarXfirst} S. Jafar ``Degrees of Freedom on the MIMO X Channel-Optimiality of the MMK Scheme," \textit{Arxiv:cs.IT/0607099v2}, September 2006.

\bibitem{Birk} Y. Birk and T. Kol, ``Informed-Source Coding-on-Demand (ISCOD) over Broadcast Channels," in \textit{Proc. IEEE INFOCOM}, San Francisco, CA, pp. 1257–1264, 1998.

\bibitem{IATerminology} S. Jafar and S. Shamai, ``Degrees of Freedom Region for the MIMO X Channel," \textit{IEEE Trans. on Information Theory}, vol. 54, no.1, pp. 151-170, January 2008.

\bibitem{BT}
G.\ Bresler and D.\ Tse, ``Degrees-of-freedom for the 3-user Gaussian Interference
Channel as a Function of Channel Diversity," \emph{Proceedings of the 2009 Allerton Conference on
Communication, Control, and Computing}, September 2009.

\bibitem{Berkeley-IA}
{\color{black}G.\ Bresler, D.\ Cartwright, and D.\ Tse, ``Settling the Feasibility of Interference Alignment for the MIMO Interference Channel: the Symmetric Case," Arxiv:1104.0888v1.}

%\bibitem{Berkeley-IA}
%G.\ Bresler, D.\ Cartwright, and D.\ Tse, ``Settling the Feasibility of Interference Alignment for the
%MIMO Interference Channel: the Symmetric Case,"
%draft manuscript, Department of EECS, University of California at Berkeley, February 2011.

%\bibitem{heath}S. Peters and R. Heath, ``Interference Alignment via Alternating Minimization,"
%in \emph{Proceedings of ICASSP 2009}, pp.\ 2445-2448, 2009.

%\bibitem {DSL}  T. Starr, J. M. Cioffi, and P. J. Silverman ``Understanding Digital Subscriber Line Technology," \textit{Prentice
%Hall}, NJ, 1999.
%\bibitem {CR}  S. Haykin ``Cognitive Radio: Brain-Empowered Wireless Communications," \textit{IEEE J. Select. Areas Commun.},  vol. 23, no. 2, pp. 201-220, Feb. 2005.
%\bibitem {AdHoc1}  A. J. Goldsmith and S. B. Wicker, ``Design Challenges for Energy-
%Constrained Ad Hoc Wireless Networks," \textit{IEEE Wireless
%Commun. Mag.},vol. 9, no. 4, pp. 8-27, Aug. 2002.
%\bibitem {AdHoc2}  I. F. Akyildiz and X. Wang, ``A Survey on Wireless Mesh Networks," \textit{IEEE Commun. Mag.},vol. 43, no. 9, pp. 23-30, Sep.
%2005.
%\bibitem{IC12} A. E. Gamal and M. H. Costa, ``The Capacity Region of a Class of Deterministic Interference Channels," \textit{IEEE Trans. on Info. Theory}, Vol.\ 33, No.\ 5, pp.\ 710-711, September 1987.
%\bibitem{IC3} T. Han and K. Kobayashi, ``A New Achievable Rate Region for the Interference
%Channel," \textit{IEEE Trans. Inf. Theory}, vol. 27, pp. 49-60, Jan.
%1981.
%\bibitem{IC1} A. B. Carleial, ``A Case Where Interference Does Not Reduce Capacity," \textit{IEEE Trans. Inf. Theory}, vol. 21, pp. 569-570, 1975.
%\bibitem{IC2} A. B. Carleial, ``Interference Channels," \textit{IEEE Trans. Inf. Theory}, vol.
%24, no. 1, pp. 60-70, 1978.
%\bibitem{IC3half} H. Sato, ``The Capacity of the Gaussian Interference Channel Under
%Strong Interference," \textit{IEEE Trans. Inf. Theory}, vol. 27, pp. 786-788,
%Nov. 1981.
%\bibitem{IC4} M. Costa, ``On the Gaussian Interference Channel," \textit{IEEE Trans. Inf.
%Theory}, vol. 31, pp. 607-615, Sep. 1985.
%\bibitem{IC5} G. Kramer, ``Feedback Strategies for White Gaussian Interference Networks,"
%\textit{IEEE Trans. Inf. Theory,} vol. 48, pp. 1423-1438, Jun. 2002.
%\bibitem{IC6} I. Sason, ``On the Achievable Rate Regions for the Gaussian Interference
%Channel," \textit{IEEE Trans. Inf. Theory}, vol. 50, pp. 1345-1356, Jun. 2004.
%\bibitem{IC7} R.\ Etkin, D.\ Tse, and H.\ Wang, ``Gaussian Interference Channel Capacity
%to within One Bit," \textit{IEEE Trans. Inf. Theory}, Dec. 2008
%\bibitem{IC8} M.\ Charafeddine, A.\ Sezgin, and A.\ Paulraj, ``Rate Region Frontiers for
%n User Interference Channel with Interference as Noise," \textit{in Proc. 45th
%Annu. Allerton Conf. Commun., Contr. Comput.}, Oct. 2007.
%\bibitem{IC9} C. Rao and B. Hassibi, ``Gaussian Interference Channel at Low SNR," \textit{in
%Proc. IEEE Int. Symp. Inf. Theory (ISIT)}, July 2004.
%\bibitem{IC10} A.\ Motahari and A.\ Khandani, ``Capacity Bounds for the Gaussian Interference
%Channel," \textit{IEEE Trans. Inf. Theory}, February 2009
%\bibitem{IC11} D.\ Tuninetti, ``Progresses on Gaussian Interference Channels with and
%without Generalized Feedback," \textit{in Proc. 2008 Inf. Theory Appl. Workshop},
%San Diego, CA, January 2008, Univ. of California.
%\bibitem{IC13} S.\ Yang and D.\ Tuninetti, ``A New Achievable Region for Interference
%Channel with Generalized Feedback," \textit{in Proc. 42nd Annu. Conf. Inf. Sci.
%Syst. (CISS)}, Mar. 2008.
%\bibitem{IC14} V.\ Annapureddy and V.\ Veeravalli, ``Gaussian Interference Networks:
%Sum Capacity in the Low Interference Regime and New Outer Bounds on
%the Capacity Region," \textit{IEEE International Symposium on Information Theory, ISIT}, 2008.
\bibitem{JafarFakhK2} S.A.\ Jafar, and M.\ Fakhereddin, ``Degrees of Freedom for the MIMO Interference Channel," \textit{IEEE Trans. on Information Theory}, Vol. 53, No. 7, pp. 2637-2642,  July 2007.
%\bibitem{HMNosratinia} A. Host-Madsen and A. Nosratinia, ``The Multiplexing Gain of Wireless Networks," \textit{in Proc. of ISIT}, 2005.
%\bibitem{RejectHMNosratinia} V.R.\ Cadambe, S.A.\ Jafar, and C.\ Wang, ``Interference Alignment with Asymmetric Complex Signaling - Settling the Host-Madsen-Nosratinia Conjecture,", Accepted for Publication in the \textit{IEEE Transactions on Information Theory}, available at \textit{arxiv:0904.0274}.
%\bibitem {Jafarbound} C.M.\ Yetis, T.\ Gou, S.A.\ Jafar, and A.H.\ Kayran, ``On Feasibility of Interference Alignment in MIMO Interference Networks," \emph{IEEE Transactions on Signal Processing}, Vol.\ 58, No.\ 9, pp.\ 4771-4782, September 2010.
\bibitem {complexity}  Z.-Q.\ Luo and S.\ Zhang, ``Dynamic Spectrum Management: Complexity and Duality," \textit{IEEE Journal of Selected Topics in Signal Processing}, Special Issue on Signal Processing and Networking for Dynamic Spectrum Access, Vol.\ 2, pp.\ 57-73, 2008.
\bibitem{Field_book}
P.\ Morandi, \emph{Field and Galois Theory}, Graduate Texts in Mathematics,  {Springer}, September 2008.
\bibitem{IANPhard} M.\ Razaviyayn, M.S.\ Boroujeni, and Z.-Q.\ Luo, ``Linear Transceiver Design for
Interference Alignment: Complexity and Computation," Available on \textit{arxiv:1009.3481}.
\bibitem{WMMSE} Q.\ Shi, M.\ Razaviyayn, Z.-Q.\ Luo and C.\ He, ``An Iteratively Weighted MMSE Approach
to Distributed Sum-Utility Maximization for a MIMO Interfering Broadcast Channel," To appear in \textit{International Conference on Acoustics, Speech, and Signal Processing (ICASSP)}, May 2011.
\bibitem{Matsumura} H.\ Matsumura, \emph{Commutative Algebra},  {Benjamin/Cummings Publications Co.}, 1980.
\bibitem{West} {\color{black}D.B.\ West, \emph{Introduction to Graph Theory}, 2nd Edition,   {Prentice Hall}, 2001.}
\bibitem{Atiyah} M. F. Atiyah and I. G. Macdonald, \emph{Introduction to Commutative Algebra}, {Westview Press}, 1994.
\bibitem{Berry} {\color{black} C. Shi, R. A. Berry, and M. L. Honig, ``Interference Alignment in Multi-Carrier Interference Networks," \textit{IEEE International Symposium on Information Theory (ISIT)}, July 2011.}
\end{thebibliography}
\end{document}